\documentclass[twocolumn,amsmath,amssymb,10pt,aps]{revtex4}
  
\usepackage[utf8]{inputenc}
\usepackage[T1]{fontenc}

\usepackage{amsmath}
\usepackage{amsfonts}
\usepackage{graphicx}
\usepackage{yfonts}
\usepackage{color}
\usepackage[normalem]{ulem}
\usepackage{amsthm}
\usepackage{bm}
\usepackage{bbm}
\usepackage{mathtools}
\usepackage{array}
\usepackage{placeins}
\usepackage{enumitem}

\def\<{\langle}
\def\>{\rangle}

\newcommand{\Tr}{\mathrm{Tr}}
\def\oper{{\mathchoice{\rm 1\mskip-4mu l}{\rm 1\mskip-4mu l}
{\rm 1\mskip-4.5mu l}{\rm 1\mskip-5mu l}}}
\DeclareMathAlphabet\mathbfcal{OMS}{cmsy}{b}{n}

\newtheorem{Theorem}{Theorem}

\newtheorem{Remark}{Remark}

\begin{document}

\title{Generalization of Pauli channels through mutually unbiased measurements}

\author{Katarzyna Siudzi{\'n}ska}
\affiliation{Institute of Physics, Faculty of Physics, Astronomy and Informatics \\  Nicolaus Copernicus University, Grudzi\k{a}dzka 5/7, 87--100 Toru{\'n}, Poland}

\begin{abstract}
We introduce a new generalization of the Pauli channels using the mutually unbiased measurement operators. The resulting channels are bistochastic but their eigenvectors are not unitary. We analyze the channel properties, such as complete positivity, entanglement breaking, and multiplicativity of maximal output purity. We illustrate our results with the maps constructed from the Gell-Mann matrices and the Heisenberg-Weyl observables.
\end{abstract}

\flushbottom

\maketitle

\thispagestyle{empty}

\section{Introduction}

The concept of mutual unbiasedness was first considered in regard to orthonormal vector bases. Two orthonormal bases are called {\it mutually unbiased} if the probability of transition between any of their vectors is constant. The $d$-dimensional Hilbert space admits at most $d+1$ mutually unbiased bases (MUBs), and the maximum is reached for $d$ being a prime power \cite{Wootters, Ivonovic}. In any dimension $d$, one can always construct at least three MUBs \cite{MUB-2}. A new approach to unbiasedness has been introduced by Kalev and Gour \cite{Kalev}, who generalized the notion of mutually unbiased bases to the mutually unbiased measurements (MUMs). These are the sets of positive operators that sum up to identity and contain the projectors onto MUB vectors as a special case of projective measurements. Interestingly, one can always construct $d+1$ MUMs, regardless of the dimension $d$.

The applications of mutually unbiased measurements have been widely studied in uncertainty relations and entanglement detection. In particular, the MUMs were used to derive state-dependent \cite{ChenFei}, state-independent \cite{Rastegin2}, and fine-grained \cite{Rastegin} entropic uncertainty relations. The last type helped to find new separability conditions for bipartite system \cite{Rastegin4}. Moreover, it was shown that there is an equality between the amounts of uncertainty for MUMs and entanglement of the measured states quantified by the conditional collision entropy \cite{Wang}. New separability criteria were given for arbitrary $d$-dimensional bipartite \cite{ChenMa,Shen,ShenLi} and multipartite systems \cite{Liu,ChenLi}. Liu, Gao, and Yan \cite{Liu2} provided the criteria whose experimental implementation does not require a full state tomography. In another paper \cite{Liu3}, they also presented the conditions for $k$-nonseparability detection of multipartite qudit systems. Graydon and Appleby generalized projective 2-designs to conical t-designs \cite{Graydon} and applied them to describe a connection between designs and entanglement \cite{Graydon2}. Recently, the MUMs have also been used to find more operational Einstein-Podolsky-Rosen (EPR) steering inequalities \cite{Lai}. Finally, Li et al. \cite{Li} used the MUMs to introduce new positive quantum maps and entanglement witnesses, which generalize the constructions from \cite{MUBs}.

In this paper, we construct a new class of bistochastic quantum channels. These channels generalize Nathanson and Ruskai's {\it diagonal Pauli channels constant on axes} (also known as {\it generalized Pauli channels}), whose definition includes the mutually unbiased bases. Our construction method uses the mutually unbiased measurement operators. It is valid for any finite dimension $d$, as one can always find the maximal number of $d+1$ MUMs. We find the necessary and sufficient conditions for complete positivity of the channels. We also analyze how the properties of Nathanson and Ruskai's channels change after replacing the MUBs with MUMs. Finally, we provide examples of bistochastic channels whose eigenvectors are not unitary operators.

\section{Mutually unbiased measurements}

Following the work by Kalev and Gour \cite{Kalev}, let us introduce the notion of mutually unbiasedness for measurement operators.
In quantum mechanics, a measurement is determined by a set of measurement operators (POVMs) $M_k$ that are positive and sum up to identity, $\{M_k|M_k\geq 0,\sum_kM_k=\mathbb{I}_d\}$. The probability of the $k$-th outcome is $\Tr(M_k\rho)$, where $\rho$ is the density operator of a quantum system. As a special case, one considers the projective measurement, where the measurement operators $M_k$ are orthogonal projectors. Clearly, mutually unbiased bases can be used to perform projective measurements. Consider $N$ orthonormal bases $\{\psi_k^{(\alpha)},k=0,\ldots,d-1\}$ in $\mathbb{C}^d$ that are numbered by $\alpha=1,\ldots,N$. These bases are mutually unbiased if and only if $|\<\psi_k^{(\alpha)}|\psi_l^{(\beta)}\>|^2=1/d$ for $\alpha\neq\beta$. Therefore, the set of projectors $P_k^{(\alpha)}=|\psi_k^{(\alpha)}\>\<\psi_k^{(\alpha)}|$ onto the $\alpha$-th basis forms a measurement with the following properties,
\begin{equation}\label{MUB}
\begin{split}
&\Tr(P_k^{(\alpha)})=1,\\
&\Tr(P_k^{(\alpha)}P_l^{(\beta)})=\delta_{\alpha\beta}\delta_{kl}+\frac 1d (1-\delta_{\alpha\beta}).
\end{split}
\end{equation}
Note that $P_k^{(\alpha)}$ can be regarded as either states or measurement operators. Therefore, two projective measurements are mutually unbiased if the probability of measuring one with the other is constant. This notion can be generalized to POVMs. Namely, the measurements $\{P_k^{(\alpha)}|P_k^{(\alpha)}\geq 0,\sum_{k=0}^{d-1}P_k^{(\alpha)}=\mathbb{I}_d\}$ are mutually unbiased if and only if \cite{Kalev}
\begin{equation}\label{MUM}
\begin{split}
&\Tr(P_k^{(\alpha)})=1,\\
&\Tr(P_k^{(\alpha)}P_l^{(\beta)})=\frac 1d +\frac{d\kappa-1}{d-1}\delta_{\alpha\beta}\left(\delta_{kl}
-\frac 1d \right),
\end{split}
\end{equation}
where $1/d<\kappa\leq 1$. For $\kappa=1$, the above conditions reproduce eq. (\ref{MUB}). It is important to note that one can always construct the maximal number of $d+1$ mutually unbiased measurements. Moreover, MUMs form an informationally complete set, and any state can be written as
\begin{equation}\label{rho}
\rho=\frac 1d \mathbb{I}_d+ \frac{d-1}{d\kappa-1}\sum_{\alpha=1}^{d+1}\sum_{k=0}^{d-1}P_k^{(\alpha)}\left[\Tr\left(\rho P_k^{(\alpha)}\right)-\frac 1d\right].
\end{equation}
In their seminal paper, Kalev and Gour \cite{Kalev} proposed a method of constructing $d+1$ MUMs from an orthonormal basis $\{\mathbb{I}_d/\sqrt{d},F_{\alpha,k}|\alpha=1,\ldots,d+1,k=1,\ldots,d-1\}$, where $F_{\alpha,k}$ are traceless Hermitian operators. Namely, one has
\begin{equation}
P_k^{(\alpha)}=\frac 1d \mathbb{I}_d+tF_k^{(\alpha)},
\end{equation}
where
\begin{equation}
F_k^{(\alpha)}=\begin{cases}
\sum_{l=1}^{d-1}F_{\alpha,l}-\sqrt{d}(1+\sqrt{d})F_{\alpha,k},\, &k\neq 0,\\
(1+\sqrt{d})\sum_{l=1}^{d-1}F_{\alpha,l},\, & k=0,
\end{cases}
\end{equation}
and $t\neq 0$ is a free real parameter such that $P_k^{(\alpha)}\geq 0$. The relation between $t$ and $\kappa$ reads
\begin{equation}\label{t}
\kappa=\frac 1d +(d-1)t^2(1+\sqrt{d})^2.
\end{equation}

\section{Generalization of Pauli channels}

A mixed unitary evolution of a qubit is described by the Pauli channel
\begin{equation}\label{Pauli}
\Lambda[\rho]=\sum_{\alpha=0}^3 p_\alpha \sigma_\alpha \rho \sigma_\alpha,
\end{equation}
which is the most general form of a bistochastic quantum channel \cite{King,Landau}. In the above formula, $p_\alpha$ denote the probability distribution, and $\sigma_\alpha$ are the Pauli matrices. One has
\begin{equation}
\Lambda[\sigma_\alpha]=\lambda_\alpha\sigma_\alpha,
\end{equation}
where $\lambda_0=1$ and
\begin{equation}
\lambda_\alpha = 2(p_0 + p_\alpha) - 1
\end{equation}
for $\alpha=1,2,3$. Now, $\Lambda$ is completely positive if and only if its eigenvalues $\lambda_\alpha$ satisfy the Fujiwara-Algoet conditions \cite{Fujiwara,King,Szarek}
\begin{equation}\label{Fuji-2}
-1 \leq \sum_{\alpha=1}^{3} \lambda_\alpha\leq 1+2\min_\alpha\lambda_\alpha.
\end{equation}

An interesting feature of the Pauli channels is that the eigenvectors of their Kraus operators $\sigma_\alpha$ are mutually unbiased. This property was used by Nathanson and Ruskai \cite{Ruskai} to introduce the generalized Pauli channels
\begin{equation}\label{GPC}
\Lambda=\frac{dp_0-1}{d-1}\oper+\frac{d}{d-1}\sum_{\alpha=1}^{d+1}p_\alpha\Phi_\alpha,
\end{equation}
where $p_\alpha$ is the probability distribution, $\oper$ denotes the identity map, and
\begin{equation}\label{qc}
\Phi_\alpha[X]=\sum_{k=0}^{d-1}P_k^{(\alpha)}\Tr(XP_k^{(\alpha)})
\end{equation}
are the quantum-classical channels constructed from the projectors $P_k^{(\alpha)}$ onto the MUB vectors. It has been shown that \cite{mub_final}
\begin{equation}\label{prop}
\begin{split}
&\Phi_\alpha\Phi_\beta=\Phi_0,\quad\alpha\neq\beta,\\
&\Phi_\alpha\Phi_\alpha=\Phi_\alpha,\\
&\sum_{\alpha=1}^{d+1}\Phi_\alpha=d\Phi_0+\oper,
\end{split}
\end{equation}
where $\Phi_0[X]=\mathbb{I}_d\Tr(X)/d$ is the completely depolarizing channel. Now, the eigenvalue equations for $\Lambda$ read $\Lambda[\mathbb{I}_d]=\mathbb{I}_d$ and
\begin{equation}
\Lambda[U_{\alpha,k}]=\lambda_\alpha U_{\alpha,k}
\end{equation}
with the unitary operators
\begin{equation}\label{UU}
U_{\alpha,k}=\sum_{l=0}^{d-1}\omega^{kl}P_l^{(\alpha)},\qquad \omega=e^{2\pi i/d}.
\end{equation}
This indicates that $\Lambda$ is a bistochastic channel. It is also self-dual ($\Lambda=\Lambda^\dagger$), so its $(d-1)$-times degenerated eigenvalues 
\begin{equation}
\lambda_\alpha=\frac{1}{d-1}\left[d(p_\alpha+p_0)-1\right]
\end{equation}
are real. Note that the generalized Pauli channels can be equivalently written as
\begin{equation}\label{GPC2}
\Lambda=p_0\oper+\frac{1}{d-1}\sum_{\alpha=1}^{d+1}p_\alpha\mathbb{U}_\alpha,
\end{equation}
where
\begin{equation}\label{U}
\mathbb{U}_\alpha[X]=\sum_{k=1}^{d-1}U_{\alpha,k}XU_{\alpha,k}^\dagger
\end{equation}
satisfies $\mathbb{U}_\alpha=d\Phi_\alpha-\oper$.
The complete positivity conditions for $\Lambda$ are the generalized Fujiwara-Algoet conditions \cite{Ruskai}
\begin{equation}\label{FA}
-\frac{1}{d-1}\leq\sum_{\alpha=1}^{d+1}\lambda_\alpha\leq 1+d\min_\alpha\lambda_\alpha.
\end{equation}

Our goal is to generalize the generalized Pauli channels from eq. (\ref{GPC}). We achieve this by replacing the mutually unbiased bases with mutually unbiased measurements. After this procedure, the bistochastic quantum-classical channels from eq. (\ref{qc}) no longer satisfy properties (\ref{prop}) but instead
\begin{equation}
\begin{split}
&\Phi_\alpha\Phi_\beta=\Phi_0,\quad\alpha\neq\beta,\\
&\Phi_\alpha\Phi_\alpha[P_l^{(\beta)}]=\Phi_0,\quad\alpha\neq\beta,\\
&\Phi_\alpha\Phi_\alpha[P_l^{(\alpha)}]=\frac{d\kappa-1}{d-1}\Phi_\alpha[P_l^{(\alpha)}]\\
&\qquad\qquad\qquad+\frac{d(1-\kappa)}{d-1}\Phi_0[P_l^{(\alpha)}],\\
&\sum_{\alpha=1}^{d+1}\Phi_\alpha=\frac{d(d-\kappa)}{d-1}\Phi_0+\frac{d\kappa-1}{d-1}\oper.
\end{split}
\end{equation}
The generalized Pauli channel has the form (\ref{GPC}), and its eigenvectors are again $U_{\alpha,k}$ defined in eq. (\ref{UU}), as
\begin{equation}
\Phi_\alpha[U_{\beta,l}]=\frac{d\kappa-1}{d-1}\delta_{\alpha\beta}U_{\beta,l}.
\end{equation}
This time, however, there is no simple relation between $\mathbb{U}_\alpha$ and $\Phi_\alpha$. Now, $U_{\alpha,k}$ form an orthogonal basis with
\begin{equation}\label{norm}
\Tr\left[U_{\alpha,k}U_{\beta,l}^\dagger\right]=\frac{d(d\kappa-1)}{d-1}\delta_{\alpha\beta}\delta_{kl}.
\end{equation}
Observe that, in terms of the Hermitian orthonormal basis, these operators read
\begin{equation}\label{Uak}
U_{\alpha,k}=\sqrt{d}t\sum_{l=1}^{d-1}F_{\alpha,l}\left[1-(\sqrt{d}+1)\omega^{kl}\right],
\end{equation}
and they are no longer unitary by definition. Therefore, our method allows one to construct bistochastic channels whose eigenvectors are not unitary. 
The eigenvalues of $\Lambda$ are given by
\begin{equation}
\lambda_\alpha=\frac{1}{d-1}\left[d\left(p_0+\frac{d\kappa-1}{d-1}p_\alpha\right)-1\right],
\end{equation}
where the inverse relation is
\begin{equation*}
p_0=\frac{1}{d^2(d-\kappa)}\left[(d-1)^2\sum_{\alpha=1}^{d+1}\lambda_\alpha-d(d\kappa-1)+d^2-1\right],
\end{equation*}
\begin{equation*}
\begin{split}
p_\alpha=\frac{(d-1)^2}{d^2(d\kappa-1)(d-\kappa)}\Bigg[d\kappa-1&+d(d-\kappa)\lambda_\alpha
\\&-(d-1)\sum_{\beta=1}^{d+1}\lambda_\beta\Bigg].
\end{split}
\end{equation*}
Sufficient complete positivity conditions for the generalized Pauli channels read $p_0\leq 1/d$ and $p_\alpha\geq 0$, which is equivalent to
\begin{equation}
\frac{d\kappa-1}{d-1}\leq\sum_{\alpha=1}^{d+1}\lambda_\alpha\leq\frac{1}{d-1}
\left[d\kappa-1+d(d-\kappa)\min_\beta\lambda_\beta\right].
\end{equation}

\begin{Remark}
Note that the operators $U_{\alpha,k}$ can be used to generate $d+1$ mutually unbiased measurements. Namely, let us take any orthogonal basis $\{\mathbb{I}_d,U_{\alpha,k}\}$ with $U_{\alpha,k}^\dagger=U_{\alpha,d-k}$, a normalization given by eq. (\ref{norm}), and satisfying an additional condition $\mathbb{I}_d+\sum_{l=1}^{d-1}\omega^{-kl}U_{\alpha,l}\geq 0$. Then, by inverting formula (\ref{UU}), we get
\begin{equation}
P_k^{(\alpha)}=\frac 1d \left[\mathbb{I}_d+\sum_{l=1}^{d-1}\omega^{-kl}U_{\alpha,l}\right].
\end{equation}
It is straightforward to check that such $P_k^{(\alpha)}$ form a POVM and satisfy conditions (\ref{MUM}). If for $U_{\alpha,k}$ one chooses the operators that form a cyclic subgroup, then $U_{\alpha,k}$ are linearly proportional to the Weyl operators $W_{kl}=\sum_{m=0}^{d-1}\omega^{mk}|m+l\>\<m|$,
and $P_k^{(\alpha)}$ are rank-1 projectors onto the mutually unbiased bases.
\end{Remark}

\section{Properties}

In this section, we show which properties of the generalized Pauli channels constructed from MUBs transfer over after one replaces MUBs with MUMs. First, in general, the generalized Pauli channels $\Lambda$ are no longer covariant with respect to all $U_{\alpha,k}$ \cite{ICQC}, which means that
\begin{equation}
\Lambda[U_{\alpha,k}XU_{\alpha,k}^\dagger]=U_{\alpha,k}\Lambda[X]U_{\alpha,k}^\dagger
\end{equation}
does not hold for an arbitrary operator $X$ \cite{Scutaru}. However, one still has $\Lambda\Phi_\alpha=\Phi_\alpha\Lambda$ for every $\alpha=1,\ldots,d+1$. Next, the necessary and sufficient conditions $\sum_{\alpha=1}^{d+1}\lambda_\alpha\leq 1$ for entanglement breaking of $\Lambda$ with $\lambda_\alpha\geq 0$ \cite{Ruskai} become only necessary. New sufficient conditions are established in the following theorem.

\begin{Theorem}\label{TH}
The generalized Pauli channel with $\lambda_\alpha\geq 0$ is entanglement breaking if $\sum_{\alpha=1}^{d+1}\lambda_\alpha\leq\frac{d\kappa-1}{d-1}$.
\end{Theorem}

\begin{proof}
Recall that a quantum channel $\Lambda$ is entanglement breaking if and only if it can be written in the Holevo form $\Lambda[X]=\sum_j R_j\Tr(\rho E_j)$, where $R_j$ are density operators and $\{E_j\}$ form a POVM \cite{EBC}. Observe that the generalized Pauli channels can be equivalently rewritten as
\begin{equation}\label{alt}
\Lambda=\left(1-\sum_{\alpha=1}^{d+1}\mu_\alpha\right)\Phi_0
+\sum_{\alpha=1}^{d+1}\mu_\alpha\Phi_\alpha
\end{equation}
with
\begin{equation}
\mu_\alpha=\frac{d-1}{d\kappa-1}\lambda_\alpha.
\end{equation}
Now, it is easy to show that
\begin{equation}
\Lambda=\sum_{\alpha=1}^{d+1}\sum_{k=0}^{d-1}P_k^{(\alpha)}
\Tr\left[X E_{k,\alpha}\right]
\end{equation}
with
\begin{equation}
E_{k,\alpha}=\frac{1-\sum_{\beta=1}^{d+1}\mu_\beta}{d(d+1)}\mathbb{I}_d
+\mu_\alpha P_k^{(\alpha)}.
\end{equation}
By definition, $P_k^{(\alpha)}$ is a density operator, and $E_{k,\alpha}$ is a sum of two positive operators, so it is also positive. Finally,
\begin{equation}
\sum_{\alpha=1}^{d+1}\sum_{k=0}^{d-1}E_{k,\alpha}=\mathbb{I}_d,
\end{equation}
and hence $\{E_{k,\alpha}\}$ form a POVM.
\end{proof}

There are several important measures that have been analyzed for Nathanson and Ruskai's generalized Pauli channels. One of them is the minimal and maximal channel fidelity $f_{\min/\max}(\Lambda)$ on pure states \cite{Raginsky,norms} that measures the distance between the input $P$ and the output state $\Lambda[P]$. Unluckily, the formulas for $f_{\min/\max}(\Lambda)$ do not carry over to the generalized Pauli channels constructed from MUMs. The reason is that the extremal fidelities for the channels constant on axes are reached on the projectors onto MUB vectors. However, the mutually unbiased measurements are not projectors, and the formula for the channel fidelity on mixed states is much more complicated than on pure states \cite{Sommers3}. Another important measure is the maximal output $p$-norm $\nu_p(\Lambda)$ that measures optimal output purity. In other words, $\nu_p(\Lambda)$ determines how close the channel output $\Lambda[\rho]$ is to a pure state. If $\Lambda$ is constructed from MUBs, the analytical formulas for $\nu_2(\Lambda)$ and $\nu_\infty(\Lambda)$ are known \cite{Ruskai,norms}. However, only the former generalizes in a straightforward manner.

\begin{Theorem}
The maximal output 2-norm of $\Lambda$ is equal to
\begin{equation}\label{nu}
\nu_2^2(\Lambda)=\frac 1d \left[1+(d\kappa-1)\max_\alpha\lambda_\alpha^2\right].
\end{equation}
\end{Theorem}

\begin{proof}
Starting from eq. (\ref{alt}), we calculate the Frobenius norm
\begin{equation*}
\begin{split}
\|\Lambda[\rho]\|_2^2&=\Tr(\Lambda[\rho]^2)\\&=\frac 1d +\frac{d-1}{d\kappa-1}
\sum_{\alpha=1}^{d+1}\lambda_\alpha^2\left[\sum_{k=0}^{d-1}\left(\Tr\rho P_k^{(\alpha)}\right)^2-\frac 1d\right].
\end{split}
\end{equation*}
Observe that $\|\Lambda[\rho]\|_2^2$ achieves its maximal value for $\rho=P_l^{(\alpha_\ast)}$, where $\lambda_{\alpha_\ast}=\sqrt{\max_\alpha\lambda_\alpha^2}$. We find that
\begin{equation}
\sum_{k=0}^{d-1}\left(\Tr P_k^{(\alpha)} P_l^{(\beta)}\right)^2=
\frac 1d \left[1+\frac{(d\kappa-1)^2}{d-1}\delta_{\alpha\beta}\right],
\end{equation}
and hence eq. (\ref{nu}) follows.
\end{proof}

By definition, the maximal output $2$-norm of $\Lambda$ is strongly multiplicative if $\nu_2(\Lambda\otimes\Omega)=\nu_2(\Lambda)\nu_2(\Omega)$ for any quantum channel $\Omega$. From the theorem by Fukuda, Holevo, and Nathanson \cite{FukHol,KingMats}, it follows that $\nu_2(\Lambda)$ is strongly multiplicative if and only if
\begin{equation}
\nu_2^2(\Lambda)=\frac 1d \left[1+(d-1)\max_\alpha|\lambda_\alpha|^2\right].
\end{equation}
Observe that $\nu_2(\Lambda)$ in eq. (\ref{nu}) satisfies the above condition only for $\kappa=1$.

\section{Examples}

\subsection{Pauli matrices}

Let us consider $d=2$ and construct the mutually unbiased measurements from the rescaled Pauli matrices $F_{1,\alpha}=\sigma_\alpha/\sqrt{2}$. Now, using eq. (\ref{Uak}), we find that the MUMs are given by
\begin{equation}
P_k^{(\alpha)}=\frac 1d \mathbb{I}_d+(-1)^k t\frac{1+\sqrt{2}}{\sqrt{2}}\sigma_\alpha,
\end{equation}
where
\begin{equation}
-\frac{\sqrt{2}-1}{\sqrt{2}}\leq t\leq\frac{1}{2+\sqrt{2}}
\end{equation}
and $t\neq 0$. Interestingly, the corresponding parameter $\kappa$ belongs to the full range $1/2<\kappa\leq 1$. The channel eigenvectors $U_{\alpha,1}$ are again rescaled Pauli matrices,
\begin{equation}
U_{\alpha,1}=P_0^{(\alpha)}-P_1^{(\alpha)}=2t\frac{1+\sqrt{2}}{\sqrt{2}}\sigma_\alpha.
\end{equation}
Therefore, regardless of the choice of $\kappa$, $\Lambda$ is the Pauli channel.

\subsection{Gell-Mann matrices}

Now, consider the higher-dimensional ($d\geq 3$) Hermitian generalization of the Pauli matrices, known as the Gell-Mann matrices. They are defined as follows,
\begin{align}
&\sigma_{kl}:=\frac{1}{\sqrt{2}}\left(|k\>\<l|+|l\>\<k|\right),\\
&\sigma_{lk}:=\frac{i}{\sqrt{2}}\left(|k\>\<l|-|l\>\<k|\right),\\
&\sigma_{kk}:=\sqrt{\frac{1}{k(k+1)}}\left(\sum_{j=0}^{k-1}|j\>\<j|-k|k\>\<k|\right),
\end{align}
for all $0\leq k<l\leq d-1$ and $0\leq k\leq d-1$, respectively. Together with $\sigma_{00}=\mathbb{I}_d/\sqrt{d}$, they form an orthonormal Hermitian operator basis. After Kalev and Gour \cite{Kalev}, we group the operators into the following sets,
\begin{align*}
&\{F_{\alpha,k}|k=1,\ldots,d-1\}=\{\sigma_{k,\alpha-1}|k\neq\alpha-1\},\\
&\{F_{d+1,k}|k=1,\ldots,d-1\}=\{\sigma_{kk}|k=1,\ldots,d-1\}.
\end{align*}
In the next step, we construct
\begin{align}
&F_0^{(d+1)}=(\sqrt{d}+1)\sum_{l=1}^{d-1}\sigma_{ll},\label{od}\\
&F_k^{(d+1)}=-\sqrt{d}(\sqrt{d}+1)\sigma_{kk}+\sum_{l=1}^{d-1}\sigma_{ll}
\end{align}
for $k=1,\ldots,d-1$, as well as
\begin{align}
&F_k^{(\alpha)}=-\sqrt{d}(\sqrt{d}+1)\sigma_{k,\alpha-1}+\sum_{l\neq\alpha-1}\sigma_{l,\alpha-1},\\
&F_{\alpha-1}^{(\alpha)}=(\sqrt{d}+1)\sum_{l\neq\alpha-1}\sigma_{l,\alpha-1}\label{do}
\end{align}
for $k\neq\alpha-1$ and $\alpha=1,\ldots,d$.
The corresponding MUMs are given by
\begin{equation}\label{P}
P_k^{(\alpha)}=\frac 1d \mathbb{I}_d+tF_k^{(\alpha)},
\end{equation}
where $t$ is a free parameter such that $P_k^{(\alpha)}\geq 0$. Let us take the optimal
\begin{equation}
t=\frac{\sqrt{2}}{d(\sqrt{d}+1)\sqrt{d-1}},
\end{equation}
which follows from the optimal $\kappa=\frac{d+2}{d^2}$ via eq. (\ref{t}) and guarantees that $P_k^{(\alpha)}\geq 0$ \cite{Kalev}. One easily finds
\begin{equation}
U_{d+1,k}=\sqrt{d}t\sum_{l=1}^{d-1}\sigma_{ll}\left[1-(\sqrt{d}+1)\omega^{kl}\right],
\end{equation}
\begin{equation}
\begin{split}
U_{\alpha,k}=&\sqrt{d}t\sum_{l=0}^{\alpha-2}\sigma_{l,\alpha-1}\left[1-(\sqrt{d}+1)\omega^{k(l+1)}\right]\\&
+\sqrt{d}t\sum_{l=\alpha}^{d-1}\sigma_{l,\alpha-1}\left[1-(\sqrt{d}+1)\omega^{kl}\right].
\end{split}
\end{equation}
For $d=3$, these operators explicitly read
\begin{equation}
U_{1,1}=U_{1,2}^\dagger=\frac{3}{2\sqrt{2}} (1+\sqrt{3})t
\begin{pmatrix}
0 & -1-i & 1-i\\
1+i & 0 & 0 \\
-1+i & 0 & 0
\end{pmatrix},
\end{equation}
\begin{equation}
U_{2,1}=U_{2,2}^\dagger=\frac{3}{2\sqrt{2}} (1+\sqrt{3})t
\begin{pmatrix}
0 & 1-i & 0\\
1-i & 0 & 1-i \\
0 & -1+i & 0
\end{pmatrix},
\end{equation}
\begin{equation}
U_{3,1}=U_{3,2}^\dagger=\frac{3}{2\sqrt{2}} (1+\sqrt{3})t
\begin{pmatrix}
0 & 0 & 1+i\\
0 & 0 & 1-i \\
1+i & 1-i & 0
\end{pmatrix},
\end{equation}
\begin{equation}
U_{4,1}=U_{4,2}^\dagger=\sqrt{3}t\,\mathrm{diag}
\begin{pmatrix}
2+\sqrt{3}-i \\
i(2+\sqrt{3}+i) \\
-(1+i)(1+\sqrt{3})
\end{pmatrix}^T.
\end{equation}
Observe that, among the above $U_{\alpha,k}$, the only operators that become unitary after rescaling are $U_{4,1}$ and $U_{4,2}$. These are also the only ones that mutually commute. Therefore, the generalized Pauli channel $\Lambda$ constructed from the Gell-Mann matrices is an example of a bistochastic channel whose eigenvectors are not unitary.

\subsection{Heisenberg-Weyl observables}

An alternative generalization of the Pauli matrices is provided by the Weyl operators
\begin{equation}
W_{kl}=\sum_{m=0}^{d-1}\omega^{km}|m+l\>\<m|.
\end{equation}
They form an orthogonal unitary operator basis, so they cannot be used to generate mutually unbiased measurements. However, it is possible to use $W_{kl}$ in order to introduce a Hermitian basis. Consider the case with $d=3$ and let us define, after \cite{Asadian}, the Heisenberg-Weyl observables
\begin{align}
&V_{kl}=\frac{1-i}{2\sqrt{d}}W_{kl}+\frac{1+i}{2\sqrt{d}}W_{kl}^\dagger,\quad k\leq l,\\
&V_{kl}=\frac{1+i}{2\sqrt{d}}W_{kl}+\frac{1-i}{2\sqrt{d}}W_{kl}^\dagger,\quad k>l.
\end{align}
Notably, such defined $V_{kl}$ are orthonormal traceless Hermitian operators. Now, in analogy to the previous example, we group $V_{kl}$ into the sets
\begin{align*}
&\{F_{\alpha,k}|k=1,\ldots,d-1\}=\{V_{k,\alpha-1}|k\neq\alpha-1\},\\
&\{F_{d+1,k}|k=1,\ldots,d-1\}=\{V_{kk}|k=1,\ldots,d-1\},
\end{align*}
so that eqs. (\ref{od}--\ref{do}) still apply if one replaces $\sigma_{kl}$ with $V_{kl}$. Numerical calculations show that the optimal values of $t$ and $\kappa$ read
\begin{equation}
t\simeq 0.112,\qquad\kappa\simeq 0.522.
\end{equation}
For comparison, the optimal $\kappa$ for the MUMs constructed from the Gell-Mann operators in $d=3$ is $\kappa=5/9\simeq 0.556$. Therefore, the Heisenberg-Weyl operators are a worse choice for the operator basis, as their optimal parameter $\kappa$ is further from its maximal value $\kappa=1$. 

The mutually unbiased measurements are again given by eq. (\ref{P}), and the associated operators $U_{\alpha,k}$ in $d=3$ read
\begin{equation}
U_{1,1}=U_{1,2}^\dagger=\sqrt{3}(1+\sqrt{3})t
\begin{pmatrix}
1 & 0 & 0\\
0 & \omega & 0 \\
0 & 0 & \omega^2
\end{pmatrix},
\end{equation}
\begin{equation}
\begin{split}
U_{2,1}=&U_{2,2}^\dagger=\frac{\sqrt{3}}{2} (1+\sqrt{3})t\\
&\times\begin{pmatrix}
0 & i\omega^2(\sqrt{3}-1) & \omega(2i-\omega)\\
-\omega^2(1+i) & 0 & -i\omega^2 \\
-i(\omega^2+\sqrt{3}) & -(\omega^2+i) & 0
\end{pmatrix},
\end{split}
\end{equation}
\begin{equation}
\begin{split}
U_{3,1}=&U_{3,2}^\dagger=\frac{\sqrt{3}}{2} (1+\sqrt{3})t\\
&\times\begin{pmatrix}
0 & -\omega^2 & i(\sqrt{3}-1)\omega^2\\
-i(\omega^2+\sqrt{3}) & 0 & -(\omega^2+\sqrt{3}) \\
(\sqrt{3}-1)\omega^2 & -i\omega^2 & 0
\end{pmatrix},
\end{split}
\end{equation}
\begin{equation}
U_{4,1}=U_{4,2}^\dagger=\frac{\sqrt{3}}{2} (1+\sqrt{3})t
\begin{pmatrix}
0 & -\omega & \sqrt{3}\omega\\
\sqrt{3}\omega^2 & 0 & -1 \\
-\omega^2 & \sqrt{3} & 0
\end{pmatrix}.
\end{equation}
Once again, there are only two operators $U_{\alpha,k}$ that become unitary after rescaling: $U_{1,1}$ and $U_{1,2}$. Moreover, $U_{1,1}$ and $U_{1,2}$ are linearly proportional to the Weyl operators $W_{10}$ and $W_{20}$, respectively. This time, however, there are two pairs of mutually commuting operators: the aforementioned $\{U_{1,1},U_{1,2}\}$, but also $\{U_{4,1},U_{4,2}\}$. 
Therefore, the generalized Pauli channel $\Lambda$ constructed from the Heisenberg-Weyl observables is another example of a bistochastic channel with non-unitary eigenvectors.

\section{Conclusions}

We introduced a new generalization of the Pauli channels, whose construction is based on the concept of mutually unbiased measurements. The resulting channels are bistochastic, but they not necessarily have unitary eigenvectors. We found  sufficient conditions for the generalized Pauli channels to be completely positive, as well as the conditions under which they break quantum entanglement. Further work is needed to establish the necessary and sufficient conditions for complete positivity. Also, we showed that their maximal output $2$-norm is strongly multiplicative if and only if the MUMs are rank-1 projectors onto the MUB vectors.  The next step is to analyze the evolution of open quantum systems provided by the generalized Pauli dynamical maps. It would be interesting to check the $\kappa$-dependence of quantum Markovianity.

\section*{Acknowledgements}

This paper was supported by the Polish National Science Centre project No. 2018/31/N/ST2/00250. The data that supports the findings of this study are available within the article.

\bibliography{C:/Users/cynda/OneDrive/Fizyka/bibliography}
\bibliographystyle{C:/Users/cynda/OneDrive/Fizyka/beztytulow2}

\end{document}